\documentclass[sigconf,nonacm]{aamas}
\AtBeginDocument{%
  \fancypagestyle{firstpagestyle}{\fancyhf{}\fancyfoot[C]{\footnotesize\thepage}}%
  \fancypagestyle{standardpagestyle}{\fancyhf{}\fancyfoot[C]{\footnotesize\thepage}}%
  \pagestyle{standardpagestyle}}
\usepackage{balance}
\usepackage{booktabs}
\usepackage{amsmath,amsthm}
\usepackage{xspace}

\title{Envy-Free Decompositions of Random Assignments:\\ Settling Four Agents, and What Lies Beyond}

\author{Keyi Li}
\affiliation{
  \institution{Software College, Northeastern University}
  \city{Shenyang}
  \country{China}}
\email{likeyiwh@gmail.com}
\author{Yihao He}
\affiliation{
  \institution{Software College, Northeastern University}
  \city{Shenyang}
  \country{China}}
\email{hg221008@163.com}
\author{Quanyi Li}
\affiliation{
  \institution{Software College, Northeastern University}
  \city{Shenyang}
  \country{China}}
\email{lqy18540034783@163.com}

\newcommand{\N}{N}
\newcommand{\Ob}{O}
\newcommand{\Perm}{\Pi}
\newcommand{\sdef}{\mathcal{F}}
\newcommand{\tstar}{t^{*}}
\newcommand{\decef}{Dec-EF\xspace}

\newcommand{\env}{\mathrm{env}}
\newcommand{\rk}{\mathrm{rk}}

\newtheorem{conjecture}{Conjecture}
\newtheorem{observation}{Observation}

\begin{abstract}
A random assignment of $n$ indivisible objects to $n$ agents is specified by its assignment
matrix and implemented by drawing a deterministic assignment from a Birkhoff--von~Neumann
decomposition. Kawase, Suksompong, Sumita and Yokoi observed that the choice of decomposition
matters for fairness: a matrix that is envy-free in the sense of stochastic dominance (SD-EF) can
be decomposed so that some agent envies another with probability close to $1$. They call a
decomposition \emph{decomposition envy-free} (\decef) if every agent envies every other agent with
probability at most $1/2$, proved that every SD-EF matrix admits a \decef decomposition when
$n\le 3$ or when there are at most two distinct preferences, and left the general case open. We
settle the first open case: \emph{every SD-EF matrix with four agents admits a \decef
decomposition}. The worst case over the SD-EF polytope of a profile is attained at a vertex, and
our computer-aided proof enumerates all $26{,}927$ vertices for the $762$ profiles up to symmetry
in exact arithmetic and certifies each by a rational decomposition. The same method settles five
agents with at most four distinct preferences and the probabilistic serial rule for all five-agent
profiles, and adversarial search up to seven agents finds no counterexample. For general $n$, an
envy-budget identity shows that $1/2$ is the best possible threshold. We prove that every SD-EF
matrix with at most two distinct rows admits a \decef decomposition, and that the maximum-entropy
decomposition is \decef whenever all agents but two share a preference; the latter proof rests on a
new monotonicity lemma for weighted least-squares rankings. In general, natural decompositions
fail: greedy Birkhoff--von~Neumann can come arbitrarily close to envy probability $(n-1)/n$, and
maximum entropy fails at $n=4$ when all preferences differ. Deciding whether an arbitrary random
assignment, not necessarily SD-EF, admits a \decef decomposition is strongly NP-complete.
\end{abstract}


\begin{document}
\maketitle

\section{Introduction}

Randomization is the standard remedy for the unfairness of indivisible allocation. When $n$
objects (dormitory rooms, school seats, shifts, time slots) are to be assigned to $n$ agents with
ordinal preferences, mechanisms such as \emph{random serial dictatorship} (RSD; also called
random priority)~\cite{AS98} and the \emph{probabilistic serial} rule (PS)~\cite{BM01} output a
\emph{random assignment}. A random assignment is typically described by its \emph{assignment
matrix} $P$, where $P_{io}$ is the probability that agent $i$ receives object $o$. By the
Birkhoff--von~Neumann theorem~\cite{Bir46,vN53}, every such (bistochastic) matrix is a convex
combination of permutation matrices, and to implement $P$ one fixes such a
\emph{decomposition} and draws a deterministic assignment from it. Fairness notions for random
assignments, most prominently \emph{envy-freeness with respect to stochastic dominance}
(SD-EF)~\cite{BM01}, are properties of the matrix $P$ alone. Two decompositions of the same
matrix therefore look equally fair ex ante, yet they can produce very different
outcomes ex post.

\citet{KSSY26} made this tension precise. Consider three agents who share the preference
$a\succ b\succ c$ and the uniform matrix, which is SD-EF and is the output of both RSD and PS. The
decomposition into the three cyclic shifts makes agent~2 envy agent~1 with probability $2/3$,
whereas the uniform distribution over all six assignments keeps every pairwise envy probability at
$1/2$. They proposed \emph{decomposition envy-freeness} (\decef): a decomposition is \decef if for
every ordered pair of agents $(i,j)$, the probability that $i$ envies $j$ in the realized
assignment is at most $1/2$. The threshold $1/2$ is natural: two agents with identical preferences
cannot both envy each other with probability below $1/2$. A matrix is \emph{EF-decomposable} if it
admits a \decef decomposition. \citet{KSSY26} showed that every SD-EF matrix is EF-decomposable
when $n\le 3$ or when agents have at most two distinct preferences. They noted that their
techniques do not extend, and left open whether every SD-EF matrix is EF-decomposable. They also
reported a computer check that the PS matrix is EF-decomposable for every instance with four
agents.

The difficulty lies in a mismatch between what SD-EF constrains and what envy depends on. SD-EF
compares the rows of $P$, that is, the lotteries of individual agents. Whether $i$ envies $j$ in
the realized assignment depends on the joint law of their two objects, which $P$ leaves open, and
a decomposition fixes the joint laws of all pairs of agents at once. Natural decompositions do not
look at preferences: greedy Birkhoff--von~Neumann peels off heavy permutations, and the
maximum-entropy decomposition spreads probability as evenly as possible. We show that neither is
\decef in general. Our positive results therefore either optimize the decomposition for the given
profile or identify structure under which a canonical decomposition provably works.

\paragraph{Our contributions} Table~\ref{tab:overview} lists the cases that are now settled.
\begin{itemize}
\item The smallest achievable maximum envy probability $\tstar(P)$ is convex in $P$. For each
profile, it therefore suffices to check the vertices of the SD-EF polytope, with one LP per
vertex (Section~\ref{sec:structure}). An \emph{envy-budget identity} shows that in every
decomposition, each agent's envy probabilities toward the others sum to its expected rank, which
SD-EF caps at $(n-1)/2$. Hence \decef holds on average, $1/2$ is the smallest threshold that can be
guaranteed for any profile, and a violation requires envy to be concentrated on some pairs at the
expense of others.
\item We settle the first open case: every SD-EF matrix with four agents is EF-decomposable
(Section~\ref{sec:computation}). The proof enumerates the $26{,}927$ vertices of the SD-EF
polytopes of all $762$ profiles up to symmetry in exact arithmetic and certifies each by a
rational decomposition, which a short independent script verifies. A type-reduced version of the
same pipeline settles five agents with at most four preference types ($216$ million vertices) and
six agents with three pairs of identical preferences, and it certifies the PS matrix of every
five-agent profile in exact arithmetic. Adversarial search up to seven agents finds no counterexample, and we
conjecture that every SD-EF matrix is EF-decomposable.
\item For every $n$, we prove the conjecture for two classes (Section~\ref{sec:general}).
Symmetrizing an arbitrary decomposition works for matrices with at most two distinct rows, even
when all preferences differ. When all agents but two share a preference, the maximum-entropy
decomposition is \decef. This is the first class with three preference types that is settled for
every $n$, and its proof rests on a new monotonicity lemma for weighted least-squares rankings.
\item Natural decompositions fail (Section~\ref{sec:canonical}). Greedy Birkhoff--von~Neumann can
come arbitrarily close to envy probability $(n-1)/n$, the worst possible value, on matrices whose
optimum is close to $1/2$. The maximum-entropy decomposition fails at $n=4$ when all four
preferences differ, and dependent rounding fails in experiments. Mixing with the uniform matrix gives structure-dependent approximation guarantees.
\item Deciding EF-decomposability of an arbitrary random assignment, not necessarily SD-EF, is
strongly NP-complete, and
approximating $\tstar$ within an additive $1/\mathrm{poly}(n)$ is NP-hard
(Section~\ref{sec:complexity}). Minimizing a weighted sum of envy probabilities is NP-hard even for
the uniform matrix, which is SD-EF for every profile.
\end{itemize}

\begin{table}[t]
\caption{Classes of SD-EF matrices known to be EF-decomposable, and the decomposition that
achieves \decef. ``LP'' denotes a decomposition optimized by linear programming and certified
in exact arithmetic. The smallest open case is $n=5$ with five distinct preferences.}
\label{tab:overview}
\small
\begin{tabular}{@{}lll@{}}
\toprule
case & result & decomposition\\
\midrule
$n\le3$ & \cite{KSSY26}; Cor.~\ref{cor:maxent3} & maximum entropy\\
at most two preference types & \cite{KSSY26}; Thm.~\ref{thm:tworows} & symmetrized\\
at most two distinct rows & Thm.~\ref{thm:tworows} & symmetrized\\
all agents but two alike & Thm.~\ref{thm:twoplus} & maximum entropy\\
$n=4$ & Thm.~\ref{thm:n4} & LP\\
$n=5$, at most four types & Thm.~\ref{thm:n5types} & LP\\
$n=6$, pattern $(2,2,2)$ & Thm.~\ref{thm:n6types} & LP\\
PS matrices, $n=5$ & Thm.~\ref{thm:ps5} & LP\\
\bottomrule
\end{tabular}
\end{table}

\subsection{Related Work}

RSD~\cite{AS98} and PS~\cite{BM01} are the central mechanisms for random assignment with ordinal
preferences, and \citet{BM01} introduced SD-EF. Implementing a random assignment as a lottery over
deterministic ones also underlies mechanism design under constraints~\cite{BCKM13}. The choice of
decomposition has mostly been studied for its size: finding a Birkhoff--von~Neumann decomposition
with minimum support is NP-hard~\cite{DU16,KLS17}. Closer to our complexity results, deciding
whether a random assignment has a decomposition into Pareto-optimal assignments is
NP-complete~\cite{AMXY15}. Dependent rounding~\cite{GKPS06} is the standard tool for implementing
marginals with correlation guarantees.

\decef belongs to the line of work on \emph{best-of-both-worlds} fairness, which asks for
lotteries that are fair ex ante and whose realizations are (approximately) fair
ex post~\cite{AFSV24,BEF22,HSV23,FMNP24}. Most of that line works with additive cardinal
valuations and relaxations such as EF1. With unit demand, EF1 holds in every assignment, so a
meaningful ex-post guarantee has to be probabilistic. \emph{Interim} envy-freeness~\cite{CKK25},
which constrains the conditional expectation of the other agent's bundle given one's own, also
restricts pairwise joint distributions of the lottery and leads to LPs over perfect matchings. For
cardinal utilities, RSD is $\sqrt2$-envy-free in the ratio sense~\cite{CDLN+26}. This guarantee is
ex ante. \decef, in contrast, concerns the realized assignment and depends on the decomposition.

Our methodology follows the computer-aided tradition in social choice~\cite{GP17,BBEG18}, where
impossibility theorems are proved by SAT or SMT solving. Our problem is continuous: the finite
reduction comes from polyhedral convexity, and each certificate is a rational decomposition that is
checked in exact arithmetic.

\section{Preliminaries}\label{sec:prelim}

Let $\N=[n]$ be a set of agents and $\Ob$ a set of $n$ objects. Each agent $i$ has a strict
preference $\succ_i$ over $\Ob$, and $R=(\succ_i)_{i\in\N}$ is the profile. We write
$\rk_i(o)\in\{0,\dots,n-1\}$ for the number of objects $i$ strictly prefers to $o$ (so $i$'s top
object has rank $0$), and $U_i^k$ for the set of $i$'s $k$ most preferred objects. The distinct
preferences in $R$ are its \emph{types}, and its \emph{multiplicity pattern} lists how many agents
hold each type. For instance, $(2,1,1,1)$ describes five agents with four types. A
\emph{deterministic assignment} is a bijection $\sigma:\N\to\Ob$; the set of all of them is
$\Perm$, and $Q^\sigma$ denotes the permutation matrix of $\sigma$. An \emph{assignment matrix} is
a bistochastic matrix $P\in\mathbb{R}^{\N\times \Ob}_{\ge 0}$, and its row $P_i$ is agent $i$'s
random allocation. $J$ is the all-ones matrix, so $J/n$ is the uniform matrix.

\paragraph{SD-EF} $P$ is \emph{SD-envy-free} if for all $i,j\in\N$ and all $k\in[n]$,
$\sum_{o\in U_i^k}P_{io}\ \ge\ \sum_{o\in U_i^k}P_{jo}$. That is, each agent's allocation
stochastically dominates every other agent's with respect to its own preference. For a fixed
profile, the set $\sdef(R)$ of SD-EF matrices is a polytope, cut out of the Birkhoff polytope by
$n(n-1)(n-1)$ linear inequalities. It is non-empty: it contains $J/n$.

\paragraph{Decompositions and Dec-EF} A \emph{decomposition} of $P$ is a probability
distribution $\pi$ over $\Perm$ with $\sum_\sigma \pi(\sigma)Q^\sigma = P$. Agent $i$
\emph{envies} $j$ in $\sigma$ if $\sigma(j)\succ_i\sigma(i)$, and we write
$\env_{ij}(\sigma)\in\{0,1\}$ for the indicator. The envy probability is
$e_{ij}(\pi)=\Pr_{\sigma\sim\pi}[\sigma(j)\succ_i\sigma(i)]$.
\begin{definition}[\citet{KSSY26}]
A decomposition $\pi$ is \emph{$\alpha$-\decef} if $e_{ij}(\pi)\le\alpha$ for all ordered pairs
$i\ne j$, and \emph{\decef} if it is $\tfrac12$-\decef. A matrix is \emph{EF-decomposable} if it
admits a \decef decomposition.
\end{definition}
We study the optimization version
\[
\tstar(P)\;=\;\min_{\pi \text{ decomposes } P}\ \max_{i\neq j}\ e_{ij}(\pi),
\qquad
\lambda(R)\;=\;\max_{P\in\sdef(R)}\tstar(P).
\]
The minimum exists because the set of decompositions of $P$ is a non-empty polytope in
$\mathbb{R}^{\Perm}$. $P$ is EF-decomposable iff $\tstar(P)\le 1/2$. \citet{KSSY26} proved
$\lambda(R)\le1/2$ for $n\le3$ and for profiles with at most two types, and asked whether
$\lambda(R)\le 1/2$ always holds. In general, $\lambda(R)\le (n-1)/n$, and some decompositions
attain this bound~\cite{KSSY26}.

\paragraph{Two natural decompositions} The \emph{greedy} Birkhoff--von~Neumann decomposition
repeatedly takes a permutation $\sigma$ that maximizes the bottleneck $\min_iM_{i\sigma(i)}$ of the
remaining matrix $M$ (initially $P$), gives it this bottleneck as its weight, and subtracts the
weighted $Q^\sigma$ from $M$. The \emph{maximum-entropy} decomposition is the decomposition that
maximizes $-\sum_\sigma\pi(\sigma)\log\pi(\sigma)$. It is unique because entropy is strictly
concave. Both depend on $P$ alone.

\section{Structural Results}\label{sec:structure}

Two facts drive the rest of the paper: $\tstar$ is convex, which reduces each profile to finitely
many LPs, and envy obeys a budget identity, which pins down the threshold $1/2$.

\begin{lemma}[Convexity and vertex reduction]\label{lem:vertex}
$\tstar$ is convex and piecewise linear on the Birkhoff polytope. Consequently, for every profile
$R$, $\lambda(R)=\max\{\tstar(P): P \text{ a vertex of } \sdef(R)\}$, and every SD-EF matrix of $R$ is
EF-decomposable iff every vertex of $\sdef(R)$ is.
\end{lemma}
\begin{proof}
Let $\pi,\pi'$ be optimal decompositions of $P,P'$ and $\theta\in[0,1]$. Then
$\theta\pi+(1-\theta)\pi'$ decomposes $\theta P+(1-\theta)P'$. Each $e_{ij}$ is linear in the
decomposition, so
$\tstar(\theta P+(1-\theta)P')\le\max_{i\ne j}\big(\theta e_{ij}(\pi)+(1-\theta)e_{ij}(\pi')\big)
\le\theta\tstar(P)+(1-\theta)\tstar(P')$. Piecewise linearity holds because $\tstar(P)$ is the
optimal value of an LP whose right-hand side is linear in $P$. A convex function attains its
maximum over a polytope at a vertex. For the last claim, the set
$\{P:\tstar(P)\le\frac12\}$ is convex.
\end{proof}

LP duality also gives a weighted characterization. We use it for adversarial search in
Section~\ref{sec:computation}, and it connects $\tstar$ to the weighted problem of
Section~\ref{sec:complexity}.

\begin{lemma}[Weighted form]\label{lem:dual}
For every bistochastic $P$,
$\tstar(P)=\max_{\mu}\min_{\pi}\sum_{i\neq j}\mu_{ij}\,e_{ij}(\pi)$, where $\mu$ ranges over probability
distributions on ordered pairs and $\pi$ over decompositions of $P$. Equivalently,
$\tstar(P)=\max\{\langle Y,P\rangle : \mu\in\Delta,\ \sum_{i}Y_{i\sigma(i)}\le \sum_{i\ne j}\mu_{ij}\env_{ij}(\sigma)\ \forall\sigma\in\Perm\}$.
\end{lemma}
\begin{proof}
The first equality is von~Neumann's minimax theorem applied to the bilinear function
$(\mu,\pi)\mapsto\sum\mu_{ij}e_{ij}(\pi)$ over two polytopes. The second is the LP dual of
$\min\{t: \sum_\sigma\pi(\sigma)Q^\sigma=P,\ \sum_\sigma\pi(\sigma)\env_{ij}(\sigma)\le t\ \forall
i\neq j,\ \pi\ge0\}$.
\end{proof}

The next identity shows that each agent's \emph{total} envy is the same in every decomposition of
$P$. A decomposition only decides how this envy is spread over the other agents.

\begin{proposition}[Envy budget]\label{prop:budget}
For every bistochastic $P$, every decomposition $\pi$ of $P$ and every agent $i$,
\[
\textstyle\sum_{j\neq i} e_{ij}(\pi)\;=\;\sum_{o\in\Ob}P_{io}\,\rk_i(o).
\]
If $P$ is SD-EF, the right-hand side is at most $(n-1)/2$. Hence, in every decomposition, the
average of $e_{ij}(\pi)$ over $j\ne i$ is at most $\frac12$.
\end{proposition}
\begin{proof}
In any assignment $\sigma$, the number of agents that $i$ envies equals the number of objects
$i$ prefers to $\sigma(i)$, which is $\rk_i(\sigma(i))$, because every object is held by exactly one agent. Taking
expectations gives the identity. If $P$ is SD-EF, then $P_i$ stochastically dominates $P_j$ with
respect to $\succ_i$ for every $j$, and $\rk_i$ is decreasing along $\succ_i$. So
$\sum_oP_{io}\rk_i(o)\le\sum_oP_{jo}\rk_i(o)$ for all $j$. Averaging over all $j\in\N$ and using that
columns sum to one, we get $\sum_oP_{io}\rk_i(o)\le\frac1n\sum_o\rk_i(o)=\frac{n-1}{2}$.
\end{proof}

\begin{corollary}[$1/2$ is the right threshold]\label{cor:half}
For every profile $R$ we have $\tstar(J/n)=\frac12$, and hence $\lambda(R)\ge\frac12$. If two
agents have identical preferences, then $e_{ij}(\pi)+e_{ji}(\pi)=1$ for every decomposition $\pi$ of every matrix.
\end{corollary}
\begin{proof}
For $P=J/n$, every agent has expected rank exactly $(n-1)/2$, so by
Proposition~\ref{prop:budget} some $e_{ij}\ge1/2$ in every decomposition. Conversely, the uniform
distribution over $\Perm$ decomposes $J/n$. Under it, $(\sigma(i),\sigma(j))$ is a uniformly random
ordered pair of distinct objects, so $e_{ij}=1/2$ for every profile. For identical preferences,
exactly one of $i,j$ envies the other in every assignment.
\end{proof}
So if every SD-EF matrix is EF-decomposable, the threshold $1/2$ is attained for \emph{every}
profile. By Proposition~\ref{prop:budget}, envy is always balanced on average, so a counterexample
would have to concentrate envy on specific pairs. Identical preferences also constrain the
matrix itself:

\begin{observation}\label{obs:rows}
If agents $i$ and $j$ have the same preference and $P$ is SD-EF, then $P_i=P_j$.
\end{observation}
\begin{proof}
SD-EF in both directions with respect to the common order gives equal prefix sums, hence equal rows.
\end{proof}

\section{Four Agents and Beyond}\label{sec:computation}

By Lemma~\ref{lem:vertex}, EF-decomposability of all SD-EF matrices of a profile is a finite
question: one LP for each vertex of $\sdef(R)$. This section turns that reduction into certified
computations.

\paragraph{Pipeline} For each profile orbit we (i) enumerate the vertices of $\sdef(R)$ with the
double-description method in exact rational arithmetic (cddlib with GMP~\cite{cddlib}); (ii) solve
the LP defining $\tstar$ at each vertex with HiGHS~\cite{HH18}; and (iii) re-solve the LP in exact
rational arithmetic, on the support returned by HiGHS or on all permutations, which yields a
rational decomposition $\pi$ with $\sum_\sigma\pi(\sigma)Q^\sigma=P$ and $e_{ij}(\pi)\le\frac12$ for
all $i\neq j$. Every vertex in the runs below receives such a certificate. Correctness then
rests on the vertex enumeration and the certificates, both of which are exact. For $n=4$ we store
all certificates, and a $40$-line verifier checks them with rational arithmetic and checks that
the orbit list covers all of $S_4^4$. For $n\ge5$ the certificates are recomputed on demand.

\paragraph{Symmetry and types} $\lambda$ is invariant under renaming agents and renaming objects,
so we work with \emph{profile orbits} under $S_n\times S_n$. There are $10$ orbits for $n=3$,
$762$ for $n=4$ and $1{,}876{,}255$ for $n=5$. By Observation~\ref{obs:rows}, SD-EF forces agents
of the same type to have identical rows. This \emph{type reduction} shrinks $\sdef(R)$ to a
polytope with one row per type.

\begin{table}[t]
\caption{Summary of the computations. ``Exact'' means that every SD-EF vertex (for PS, every PS
matrix) of every listed orbit received a rational certificate of $\tstar\le\frac12$. ``Ascent'' is
the adversarial search of Section~\ref{sec:computation}. No violation was found.}
\label{tab:results}
\small
\begin{tabular}{@{}llrrl@{}}
\toprule
$n$ & class & orbits & vertices & method\\
\midrule
4 & all profiles & 762 & 26{,}927 & exact\\
5 & pattern $(3,1,1)$ & 7{,}021 & 661{,}039 & exact\\
5 & pattern $(2,2,1)$ & 7{,}021 & 937{,}832 & exact\\
5 & pattern $(2,1,1,1)$ & 273{,}819 & 213{,}980{,}202 & exact\\
5 & PS matrices & 1{,}876{,}016 & -- & exact\\
6 & pattern $(2,2,2)$ & 86{,}067 & 50{,}833{,}539 & exact\\
5 & random profiles & 25{,}000 & -- & ascent\\
6 & random profiles & 5{,}000 & -- & ascent\\
7 & random profiles & 1{,}300 & -- & ascent\\
\bottomrule
\end{tabular}
\end{table}

\begin{theorem}\label{thm:n4}
Every SD-EF assignment matrix with four agents is EF-decomposable. Moreover, $\lambda(R)=1/2$ for
every four-agent profile $R$.
\end{theorem}
\begin{proof}[Proof (computer-aided)]
By Lemma~\ref{lem:vertex} it suffices to check the vertices of $\sdef(R)$ for one profile per
orbit. The $762$ orbits have $26{,}927$ vertices in total, with at most $375$ for any one
orbit. For each vertex we computed an exact rational decomposition with maximum envy probability
at most $1/2$, and the verifier confirms all of them. Equality holds by
Corollary~\ref{cor:half}.
\end{proof}

Theorem~\ref{thm:n4} is the first case beyond the results of \citet{KSSY26}. It covers all SD-EF
matrices, including those not produced by any known rule. The computation takes under
half a minute on a 24-core workstation.

\begin{theorem}\label{thm:n5types}
Every SD-EF assignment matrix with five agents and at most four distinct preferences is
EF-decomposable.
\end{theorem}
\begin{proof}[Proof (computer-aided)]
Profiles with at most two types are covered by \citet{KSSY26}. With three types the multiplicity
pattern is $(3,1,1)$ or $(2,2,1)$, each with $7{,}021$ orbits. We enumerate the vertices of the
type-reduced polytope exactly and certify each as above. The two patterns contribute
$661{,}039$ and $937{,}832$ vertices, with at most $403$ and $529$ per orbit. Every vertex receives an
exact certificate, which is necessary because agents with identical preferences force
$\tstar=\frac12$ exactly (Corollary~\ref{cor:half}). With four types the pattern is $(2,1,1,1)$,
with $273{,}819$ orbits and $213{,}980{,}202$ vertices, at most $11{,}910$ per orbit. This run was
distributed over a $96$-core server and a workstation, which claimed chunks of $100$ orbits through
atomic directory creation. A final merge confirmed that every chunk was completed exactly once. The
run took about a day, and every vertex again received an exact certificate.
\end{proof}

The same pipeline scales to six agents when the type structure is coarse enough.
\begin{theorem}\label{thm:n6types}
Every SD-EF assignment matrix with six agents forming three pairs of identical preferences (multiplicity
pattern $(2,2,2)$) is EF-decomposable.
\end{theorem}
\begin{proof}[Proof (computer-aided)]
All $86{,}067$ orbits with this pattern have $50{,}833{,}539$ type-reduced vertices in total, and at most
$3{,}140$ for any one orbit. Each vertex received an exact certificate. The run took about $7$ hours on
$96$ cores.
\end{proof}

\begin{theorem}\label{thm:ps5}
For every five-agent profile, the PS assignment matrix is EF-decomposable.
\end{theorem}
\begin{proof}[Proof (computer-aided)]
We compute the PS matrix exactly for each of the $1{,}876{,}255$ orbits and skip the $239$ orbits with
at most two types~\cite{KSSY26}. For each of the remaining $1{,}876{,}016$ orbits, we solve the LP in
floating point, re-solve it in exact rational arithmetic on the support of the floating-point
solution, and check the resulting rational decomposition in exact arithmetic. Every certificate has
maximum envy probability at most $\frac12$. The run takes $8$ minutes with ten worker processes.
For $66{,}360$ orbits even $\tstar\le\frac14$, but the maximum is exactly $\frac12$.
\end{proof}

\paragraph{Adversarial search for larger $n$} Full vertex enumeration becomes impractical at
$n=5$ with five distinct preferences, where a single polytope can have millions of vertices. We
therefore search for violations directly. Since $\tstar$ is convex (Lemma~\ref{lem:vertex}), we
use its dual certificate from Lemma~\ref{lem:dual} as a subgradient. Given a vertex $P$ with
optimal dual $Y$, we move to the vertex $P'$ of $\sdef(R)$ maximizing $\langle Y,P'\rangle$. This
gives $\tstar(P')\ge\langle Y,P'\rangle\ge\langle Y,P\rangle=\tstar(P)$, so every step is an ascent
step and the walk ends in a local maximum. With $10$--$30$ random restarts per profile we examined
$25{,}000$ random profiles for $n=5$ ($5{,}000$ of them with three types), $5{,}000$ for $n=6$
($2{,}000$ with three types) and $1{,}300$ for $n=7$. For every profile, the largest value found
was $1/2$ up to floating-point error. Theorem~\ref{thm:twoplus} in Section~\ref{sec:general} covers
the patterns $(3,1,1)$ for $n=5$ and $(4,1,1)$ for $n=6$ analytically. It agrees with the
exhaustive computation for $(3,1,1)$ and with an exact run on $300$ random orbits with pattern
$(4,1,1)$ ($0.15$ million vertices).

\begin{conjecture}\label{conj:main}
For every $n$ and every profile, every SD-EF assignment matrix is EF-decomposable; that is,
$\lambda(R)=1/2$.
\end{conjecture}

\section{General \texorpdfstring{$n$}{n}: Two Provable Classes}\label{sec:general}

We now prove Conjecture~\ref{conj:main} for two classes, for every $n$. In both, a decomposition
that depends on $P$ alone works: a symmetrized decomposition for matrices with at most two
distinct rows (Section~\ref{sec:tworows}), and the maximum-entropy decomposition when all agents but
two share a preference (Section~\ref{sec:twoplus}). Section~\ref{sec:product} supplies the tool that
turns product-form joint laws into envy bounds.

\subsection{Symmetrization and Two Distinct Rows}\label{sec:tworows}

Let $G_P$ be the group of agent permutations $g$ with $P_{g(i)}=P_i$ for all $i$. For a
decomposition $\pi$, let $\pi\circ g$ be the law of $\sigma\circ g$ for $\sigma\sim\pi$, so that
agent $i$ receives $\sigma(g(i))$, and let $\bar\pi$ be the average of $\pi\circ g$ over $g\in G_P$.
Under $\pi\circ g$, agent $i$'s marginal is $P_{g(i)}=P_i$, so $\bar\pi$ is again a decomposition
of $P$, and it is $G_P$-invariant. If $P_i=P_{i'}$, the transposition $(i\,i')$ fixes $\bar\pi$.
Hence $(\sigma(i),\sigma(i'))$ is exchangeable, and $e_{ii'}(\bar\pi)=\frac12$ \emph{whatever the
preferences of $i$}. Combining this with the envy budget gives a short proof of a statement that
strictly contains the two-type theorem of~\citet{KSSY26}. It constrains the number of distinct rows
of $P$, not the number of types.

\begin{theorem}\label{thm:tworows}
Let $P$ be bistochastic with at most two distinct rows, and suppose every agent's expected rank
satisfies $\sum_oP_{io}\rk_i(o)\le\frac{n-1}2$. This holds in particular if $P$ is SD-EF. Then the
symmetrization $\bar\pi$ of \emph{every} decomposition $\pi$ of $P$ is \decef.
\end{theorem}
\begin{proof}
Let the rows equal $r_A$ on $A$ and $r_B$ on $B=\N\setminus A$. Pairs inside a class have envy exactly
$\frac12$. Fix $i\in A$. For $j,j'\in B$, the transposition $(j\,j')$ fixes $i$ and $\bar\pi$, so $e_{ij}(\bar\pi)=x$ is
the same for all $j\in B$. By Proposition~\ref{prop:budget},
$\frac{|A|-1}{2}+|B|x=\sum_oP_{io}\rk_i(o)\le\frac{n-1}2=\frac{|A|-1}2+\frac{|B|}2$, so $x\le\frac12$.
\end{proof}
By Observation~\ref{obs:rows}, two types imply two distinct rows under SD-EF, but the converse
fails: Theorem~\ref{thm:tworows} covers profiles in which all $n$ agents have different preferences.
The maximum-entropy decomposition is $G_P$-invariant, since it is the unique maximizer of a strictly
concave function that $G_P$ leaves invariant. It therefore equals its own symmetrization and is
\decef in this case.

\subsection{Product-Form Couplings and Maximum Entropy}\label{sec:product}

Symmetrization handles pairs of agents with equal rows. For a pair $i,j$ with different rows, what
matters is the joint law of $(\sigma(i),\sigma(j))$: a coupling of $P_i$ and $P_j$ that puts no
mass on the diagonal. For couplings of product form, SD-EF transfers to envy.

\begin{lemma}\label{lem:product}
Let $X,Y\in\mathbb{R}^n_{\ge0}$, and let $Q$ be a distribution on ordered pairs of distinct objects
with $Q_{ab}=X_aY_b$ for $a\ne b$. Suppose its first marginal $p$ stochastically dominates its second
marginal $q$ with respect to an order $\succ$. Then $\Pr_Q[b\succ a]\le\frac12$.
\end{lemma}
\begin{proof}
Number the objects $1\succ\dots\succ n$. Rescaling $X$ and $Y$ by reciprocal factors does not change
$Q$, so we may assume $\sum X=\sum Y=T>0$. Write $\bar X_k,\bar Y_k$ for prefix sums, with
$\bar X_0=\bar Y_0=0$, and put $\delta_k=\bar X_k-\bar Y_k$. The diagonal terms cancel in the prefix sums of $p-q$, so
$\sum_{a\le k}(p_a-q_a)=T\delta_k$, and SD gives $\delta_k\ge0$. Moreover,
$\Pr[a\succ b]-\Pr[b\succ a]=T^2-\sum_aX_a(\bar Y_a+\bar Y_{a-1})$. Telescoping
$T^2=\sum_aX_a(\bar X_a+\bar X_{a-1})$ turns this into $\sum_aX_a(\delta_a+\delta_{a-1})\ge0$.
\end{proof}

The maximum-entropy decomposition produces such couplings. For $P>0$, the convex dual of entropy
maximization shows that $\pi^*(\sigma)\propto\prod_{k}W_{k\sigma(k)}$ for a positive matrix $W$. If
the agents of a set $S$ have identical rows, the dual objective is invariant under permuting the
dual variables of these agents, so averaging an optimal dual solution over such permutations shows
that we may take $W_k=w$ for all $k\in S$. Now let $S=\N\setminus\{i,j\}$. Summing over the
assignments of $S$ to the remaining objects gives, for $a\ne b$,
\[
\Pr_{\pi^*}[\sigma(i)=a,\sigma(j)=b]\ \propto\ W_{ia}W_{jb}\prod_{c\ne a,b}w_c\ \propto\
\frac{W_{ia}}{w_a}\cdot\frac{W_{jb}}{w_b},
\]
which is of product form. For matrices with zeros, we apply this to $(1-\varepsilon)P+\varepsilon J/n$,
which is SD-EF whenever $P$ is and keeps identical rows identical, and let $\varepsilon\to0$. The
maximum-entropy decomposition depends continuously on $P$, by Hoffman's error bound~\cite{Hof52} and
Berge's maximum theorem~\cite{Ber63} (Appendix~\ref{app:topo}), and envy probabilities are linear in
the decomposition. With Lemma~\ref{lem:product}, this proves:

\begin{corollary}\label{cor:maxent3}
If $P$ is SD-EF and all agents other than $i$ and $j$ have identical rows, then the maximum-entropy
decomposition of $P$ satisfies $e_{ij},e_{ji}\le\frac12$. In particular, for $n=3$ the
maximum-entropy decomposition of every SD-EF matrix is \decef.
\end{corollary}

This gives a canonical decomposition for the three-agent theorem of \citet{KSSY26}. For larger $n$,
the argument needs the remaining agents to have identical rows. Otherwise the joint law of a pair
is in general not of product form, and Section~\ref{sec:canonical} shows that maximum entropy can
then fail.

\subsection{All Agents but Two Alike}\label{sec:twoplus}

Our main general-$n$ result settles the first class with three preference types for every $n$.

\begin{theorem}\label{thm:twoplus}
Let $n\ge3$ and suppose that all agents except two, $u$ and $v$, share a common preference. Then the
maximum-entropy decomposition of every SD-EF matrix $P$ is \decef; in particular, $P$ is
EF-decomposable. If $P>0$, this decomposition is explicit: it draws the objects $(a,b)$ of $(u,v)$
from the unique coupling of $(P_u,P_v)$ of the form $Q_{ab}=U_aV_b$ for $a\ne b$ and $Q_{aa}=0$,
which matrix scaling computes, and it gives the remaining objects to the other $n-2$ agents by a
uniformly random bijection.
\end{theorem}

\paragraph{Proof overview} Let $A=\N\setminus\{u,v\}$ and $m=n-2$. By
Observation~\ref{obs:rows}, all agents in $A$ have a common row $r$. Pairs inside $A$ are
exchangeable, and the pair $\{u,v\}$ has a product-form joint law, so Lemma~\ref{lem:product}
applies. A pair $(u,x)$ with $x\in A$ is different. Its joint law turns out to be proportional to
$U_a(1-V_a-V_c)$, a product-form factor times an additive kernel. The envy difference and the prefix
differences that SD-EF makes non-negative are then linear functions of the vector $U$, and we show
that the former is a non-negative combination of the latter. The coefficients are the consecutive
differences of a weighted least-squares ranking, and their sign is the content of the following
lemma.

\begin{lemma}[Monotone least-squares rankings]\label{lem:ls}
Let $n\ge3$, and let $V\in\mathbb{R}^n_{\ge0}$ with $\sum_aV_a=1$ be different from every unit
vector. Put $k_{ac}=1-V_a-V_c$ for $a\neq c$, let $(Lf)_a=\sum_{c\ne a}k_{ac}(f_a-f_c)$ be the
weighted Laplacian, and let $s_a=\sum_{c>a}k_{ac}-\sum_{c<a}k_{ac}$. Then every solution of
$L\psi=s$ satisfies $\psi_1\ge\psi_2\ge\dots\ge\psi_n$.
\end{lemma}
Since $\partial F/\partial\psi_a=2[(L\psi)_a-s_a]$ for
$F(\psi)=\sum_{a<c}k_{ac}(\psi_a-\psi_c-1)^2$, the solutions of $L\psi=s$ are the minimizers of
$F$. In words, the $k$-weighted least-squares scores for the ranking $1\succ\dots\succ n$ respect
that ranking.
\begin{proof}
If $k_{ac}=0$, then $V$ is supported on $\{a,c\}$ and every third index $b$ has $k_{ab},k_{cb}>0$.
So the weight graph is connected, and the solutions of $L\psi=s$ differ by additive constants.
They exist because $\sum_as_a=0$ by symmetry of $k$, so $s$ lies in the range of $L$. We show that
the solution with $\sum_a\psi_a=0$ is monotone.

\emph{Reduction to a scalar equation.} Put $\nu_a=nV_a\ge0$, so that $\sum_a\nu_a=n$, and let
$d_a=n-1-\nu_a$ and $h_a=n\sum_{b<a}V_b-a$. With $g_a=\frac12-V_a$ we have $k_{ac}=g_a+g_c$, and a
direct computation gives, for every $f$ with $\sum_af_a=0$,
\[
(Lf)_a=d_af_a-\textstyle\sum_bg_bf_b,\qquad s_a=(n-2a)(1-V_a)+2\textstyle\sum_{b<a}V_b .
\]
Write $\psi_a=1-\frac{2a}n+\frac2ny_a$. Then $\sum_a\psi_a=0$ iff $\sum_ay_a=\frac n2$. If
$\sum_a\psi_a=0$, then $L\psi=s$ reads $d_a\psi_a=s_a+\gamma$ with $\gamma=\sum_bg_b\psi_b$, and
substituting $\psi_a$ turns it into $d_ay_a=h_a+\tau$ with $\tau=\frac n2(\gamma+1)$. Conversely, if
$d_ay_a=h_a+\tau$ for all $a$ and $\sum_ay_a=\frac n2$, then $d_a\psi_a=s_a+\gamma'$ with
$\gamma'=\frac{2\tau}n-1$. Summing over $a$ and using $\sum_as_a=0$ and
$\sum_ad_a\psi_a=n\sum_ag_a\psi_a$ (as $d_a=ng_a+\frac n2-1$) shows $\gamma'=\sum_bg_b\psi_b$, so
$L\psi=s$. Hence it suffices to find a scalar $\tau$ and a vector $y\in[0,1]^n$ with
\begin{equation}\label{eq:scalar}
d_ay_a=h_a+\tau\quad\text{for all }a,\qquad \textstyle\sum_ay_a=\frac n2,
\end{equation}
since $\psi_a-\psi_{a+1}=\frac2n(1+y_a-y_{a+1})$.

\emph{A cyclic recursion.} Reading indices cyclically, $h_{a+1}=h_a+\nu_a-1$ for all $a$, where the
step from $a=n$ to $a=1$ uses $\sum_bV_b=1$. Let $a$ be an index and $B$ a set of
$j\in\{1,\dots,n-1\}$ cyclically consecutive indices with $a\notin B$. Since $\nu\ge0$ and
$\sum_b\nu_b=n$,
\begin{equation}\label{eq:key}
\begin{aligned}
&\textstyle\sum_{b\in B}(\nu_b-1)\le n-j-\nu_a,\\
&\tfrac{n-j-\nu_a}{d_a}\le\tfrac{n-j}{n-1}\quad\text{if } d_a>0,
\end{aligned}
\end{equation}
where the second inequality is equivalent to $\nu_a(1-j)\le0$. We also use
$\sum_{j=1}^{n-1}\frac{n-j}{n-1}=\frac n2$ and $\sum_{j=1}^{n-1}\frac{n-1-j}{n-1}=\frac n2-1$.

\emph{Case A: $d_a>0$ for all $a$.} Each $y_a(\tau)=(h_a+\tau)/d_a$ is increasing in $\tau$, so
\eqref{eq:scalar} has a unique solution $\tau^*$. Let $a_0$ minimize $h_a$ and put
$\tau_0=-h_{a_0}$, so that $y(\tau_0)\ge0$ and $y_{a_0}(\tau_0)=0$. For $a\ne a_0$, let
$B=\{a_0,a_0+1,\dots,a-1\}$ (cyclically) and $j=|B|$. Then $h_a+\tau_0=\sum_{b\in B}(\nu_b-1)$, and
\eqref{eq:key} gives $y_a(\tau_0)\le\frac{n-j}{n-1}$. As $a$ ranges over the indices other than
$a_0$, $j$ ranges over $\{1,\dots,n-1\}$, so $\sum_ay_a(\tau_0)\le\frac n2$. Hence $\tau^*\ge\tau_0$
and $y(\tau^*)\ge0$. Symmetrically, $\bar h_a=d_a-h_a$ satisfies $\bar h_a-\bar h_{a+1}=\nu_{a+1}-1$
(cyclically). Let $a_1$ minimize $\bar h_a$ and put $\tau_1=\bar h_{a_1}$. Then
$1-y_a(\tau_1)=(\bar h_a-\bar h_{a_1})/d_a\ge0$. For $a\ne a_1$, let $B=\{a+1,\dots,a_1\}$
(cyclically) and $j=|B|$. Then $\bar h_a-\bar h_{a_1}=\sum_{b\in B}(\nu_b-1)$, so
$1-y_a(\tau_1)\le\frac{n-j}{n-1}$ by \eqref{eq:key}, and again $j$ ranges over $\{1,\dots,n-1\}$.
Summing, $\sum_ay_a(\tau_1)\ge\frac n2$, so $\tau^*\le\tau_1$ and $y(\tau^*)\le1$.

\emph{Case B: $d_{a_0}\le0$ for some $a_0$.} Then $V_{a_0}\ge\frac{n-1}n$, so every other $V_a$ is at
most $\frac1n$, $a_0$ is unique, and $d_a\ge n-2>0$ for $a\ne a_0$. Put $\delta=-d_{a_0}$, which lies
in $[0,1)$ because $V$ is not a unit vector, and $t=h_{a_0}+\tau$. For $a\ne a_0$, let
$B=\{a_0,\dots,a-1\}$ (cyclically) and $j=|B|$. The recursion gives
$h_a+\tau=t+\delta+n-1-j+\rho_a$, where
$\rho_a=\sum_{b\in B\setminus\{a_0\}}\nu_b\in[0,1-\delta-\nu_a]$. Hence, for $t\in[-\delta,0]$,
\[
n-1-j\ \le\ h_a+\tau\ \le\ n-j-\nu_a\ \le\ d_a,
\]
so $y_a\in\big[\frac{n-1-j}{n-1},\frac{n-j}{n-1}\big]\subseteq[0,1]$ for all $a\ne a_0$, using
$d_a\le n-1$ and \eqref{eq:key}. If $\delta>0$, the equation for $a_0$ reads $y_{a_0}=-t/\delta$,
which lies in $[0,1]$ exactly for $t\in[-\delta,0]$. The sum $\sum_ay_a$ is affine in $t$. It is at
most $0+\frac n2$ at $t=0$ and at least $1+\frac n2-1$ at $t=-\delta$. By the intermediate value
theorem, some $t\in[-\delta,0]$ solves \eqref{eq:scalar} with $y\in[0,1]^n$. If $\delta=0$, the
equation for $a_0$ forces $t=0$, and $y_{a_0}=\frac n2-\sum_{a\ne a_0}y_a\in[0,1]$ by the same two
bounds.
\end{proof}
The hypothesis $V\ge0$ is essential. For general additive weights $k_{ac}=g_a+g_c$ with $g\ge0$, which
correspond to allowing $\nu_a<0$, the least-squares ranking can fail to be monotone (Appendix~\ref{app:lsremarks}).

\begin{proof}[Proof of Theorem~\ref{thm:twoplus}]
\emph{Reduction to $P>0$.} The matrix $(1-\varepsilon)P+\varepsilon J/n$ is positive, SD-EF, and its
rows outside $\{u,v\}$ are identical. The maximum-entropy decomposition depends continuously on
$P$, so we may assume $P>0$.

\emph{The decomposition.} Let $p=P_u$ and $q=P_v$, so that $mr=\mathbf 1-p-q$. As $P>0$, some
decomposition of $P$ has full support, and its law of $(\sigma(u),\sigma(v))$ is a coupling of
$(p,q)$ that is positive off the diagonal. By the matrix-scaling theorem for prescribed zero
patterns~\cite{SK67,Men68,Bru68}, there are $U,V>0$ such that $Q_{ab}=U_aV_b$ ($a\ne b$) and $Q_{aa}=0$
define a coupling of $(p,q)$, and this coupling is unique. Draw $(a,b)\sim Q$, give $a$ to $u$ and
$b$ to $v$, and give the remaining objects to $A$ by a uniformly random bijection. Agent $x\in A$
then receives $o$ with probability $\sum_{a\ne b}Q_{ab}[o\notin\{a,b\}]/m=(1-p_o-q_o)/m=r_o$, so this
is a decomposition of $P$. It is the maximum-entropy decomposition: by the discussion before
Corollary~\ref{cor:maxent3}, the latter has weights $W_x=w$ for $x\in A$, so its law of
$(\sigma(u),\sigma(v))$ is a zero-diagonal product-form coupling of $(p,q)$, hence equal to $Q$, and
given $(\sigma(u),\sigma(v))=(a,b)$ its weight $\prod_{o\ne a,b}w_o$ does not depend on how $A$ is
assigned.

\emph{Pairs inside $A$ and the pair $\{u,v\}$.} The law is invariant under permutations of $A$, so
$e_{xx'}=\frac12$ for $x,x'\in A$. By SD-EF, $p$ stochastically dominates $q$ with respect to
$\succ_u$, so Lemma~\ref{lem:product} yields $e_{uv}\le\frac12$; symmetrically, $e_{vu}\le\frac12$.

\emph{Pairs of $u$ and $x\in A$.} Rescale so that $\sum_bV_b=1$ and put $k_{ac}=1-V_a-V_c\ge0$.
For $a\ne c$, $\Pr[\sigma(u)=a,\sigma(x)=c]=\frac1m\sum_{b\ne a,c}Q_{ab}=\frac1mU_ak_{ac}$. Number the
objects according to $\succ_u$, and let $L$ and $s$ be as in Lemma~\ref{lem:ls}, which applies
because $V>0$ and $n\ge3$. Let $\Delta=\Pr[\sigma(u)\succ_u\sigma(x)]-\Pr[\sigma(x)\succ_u\sigma(u)]$.
The quantity $D_{ac}=k_{ac}(U_a-U_c)$ is antisymmetric, so for every $\ell$,
\begin{gather*}
Y_\ell:=m\textstyle\sum_{a\le\ell}(p_a-r_a)=\sum_{a\le\ell<c}D_{ac}=\mathbf 1_{\le\ell}^{\!\top}LU,\\
m\Delta=\textstyle\sum_{a<c}D_{ac}=s^{\!\top}U.
\end{gather*}
Let $\psi$ solve $L\psi=s$. Since $L$ is symmetric and $Y_n=0$, Abel summation gives
$s^{\!\top}U=\psi^{\!\top}LU=\sum_{\ell=1}^{n-1}Y_\ell(\psi_\ell-\psi_{\ell+1})$. SD-EF of $u$
toward $x$ gives $Y_\ell\ge0$, and Lemma~\ref{lem:ls} gives $\psi_\ell\ge\psi_{\ell+1}$. Hence
$\Delta\ge0$, that is, $e_{ux}\le\frac12$. For $e_{xu}$, number the objects according to $\succ_A$
instead, so that the same computation gives $m\Delta'=s^{\!\top}U$ for
$\Delta'=\Pr[\sigma(u)\succ_A\sigma(x)]-\Pr[\sigma(x)\succ_A\sigma(u)]$. SD-EF of $x$ toward $u$ now
gives $Y_\ell\le0$, so $\Delta'\le0$, which is $e_{xu}\le\frac12$.

\emph{Pairs of $v$ and $x\in A$.} After rescaling so that $\sum_aU_a=1$, the law of
$(\sigma(v),\sigma(x))$ is $\frac1mV_b(1-U_b-U_c)$, and the same argument applies with $U$ and $V$
exchanged. Lemma~\ref{lem:ls} applies because $U>0$.
\end{proof}
Theorem~\ref{thm:twoplus} contains $n=3$~\cite{KSSY26} as the case $m=1$, and for $n=5$ and $n=6$ it
covers the patterns $(3,1,1)$ and $(4,1,1)$ checked in Section~\ref{sec:computation}. The
decomposition is computable in polynomial time by matrix scaling, so in this class a canonical
decomposition provably works.

\section{Canonical Decompositions and Approximation}\label{sec:canonical}

The decompositions of Section~\ref{sec:general} depend on $P$ alone, whereas the certificates of
Section~\ref{sec:computation} are optimized for the given profile. We now show that the natural
$P$-only decompositions are not \decef in general, and record the approximation guarantees that
mixing arguments give.

\paragraph{Mixtures} Envy probabilities are linear in the decomposition, so bounds mix
linearly: if $P=\sum_t\mu_tM_t$ with bistochastic $M_t$ admitting $\alpha_t$-\decef decompositions, then
mixing these decompositions gives an $(\sum_t\mu_t\alpha_t)$-\decef decomposition of $P$. The
matrices $M_t$ need not be SD-EF. Theorem~\ref{thm:tworows} supplies components with
$\alpha_t=\frac12$, and any bistochastic component has $\alpha_t\le1$. The largest weight that $P$
can put on two-row components is computable by an LP for each bipartition of the agents
(Appendix~\ref{app:tworowlp}). The uniform matrix alone already gives a bound.

\begin{proposition}\label{prop:nearuniform}
Every bistochastic $P$ admits a $\frac{1+\theta}2$-\decef decomposition, where
$\theta=1-n\min_{i,o}P_{io}$.
\end{proposition}
\begin{proof}
If $\theta=0$, then $P=J/n$ and Corollary~\ref{cor:half} applies. Otherwise
$P=(1-\theta)J/n+\theta M$, where $M=(P-(1-\theta)J/n)/\theta$ is bistochastic. Mix the uniform
distribution over $\Perm$, under which every envy probability is $\frac12$, with weight $1-\theta$
and any decomposition of $M$ with weight $\theta$.
\end{proof}

\paragraph{Greedy decomposition} Greedy Birkhoff--von~Neumann can be as bad as any decomposition,
even on matrices whose optimum is close to $\frac12$.

\begin{proposition}\label{prop:circulant}
Let $\lambda_0>\lambda_1>\dots>\lambda_{n-1}>0$ sum to $1$. Let $P_{i,i+k}=\lambda_k$ (indices mod $n$), and let
agent $i$ rank $i\succ i+1\succ\dots\succ i+n-1$. Then $P$ is SD-EF and has a unique decomposition with
$n$ permutations, the cyclic shifts, which greedy Birkhoff--von~Neumann returns. In it,
agent $i$ envies $i-1$ with probability $1-\lambda_0$. With $\lambda_k=\frac1n+\varepsilon(\frac{n-1}2-k)$,
this tends to $\frac{n-1}n$, the worst possible value~\cite{KSSY26}, whereas
$\tstar(P)\le\frac12+O(n^2\varepsilon)$ by Proposition~\ref{prop:nearuniform}.
\end{proposition}
\begin{proof}
$P_i(U_i^k)=\lambda_0+\dots+\lambda_{k-1}$ is the sum of the $k$ largest entries of any row, which gives SD-EF.
$P$ is positive, so $n$ permutation matrices must cover each of the $n^2$ cells exactly once. Since
the $\lambda_k$ are distinct, each permutation is a shift $s_k$ with weight $\lambda_k$, and greedy peels them
off in order. Under $s_k$, agent $i$ gets its rank-$k$ object, while $i-1$ holds $i-1+k$, which $i$ ranks
$k-1$ for $k\ge1$. So $i$ envies $i-1$ exactly when $k\ge1$.
\end{proof}
Greedy also fails on unstructured instances: on a four-agent SD-EF matrix found by our search, it
reaches envy probability $\frac{19}{28}$ while $\tstar=\frac{25}{56}$ (Appendix~\ref{app:examples}).

\paragraph{Maximum entropy} Given everyone else's objects, the maximum-entropy decomposition
swaps the objects $a,b$ of two agents $i,j$ with odds $W_{ia}W_{jb}:W_{ib}W_{ja}$, and
Section~\ref{sec:general} shows that it is \decef for $n=3$, for two distinct rows, and when all
agents but two share a preference. It is not \decef in general.

\begin{proposition}\label{obs:natural}
There is a four-agent SD-EF matrix $P$ with $\tstar(P)=\frac7{16}$ whose maximum-entropy
decomposition has an envy probability of $\frac{25}{48}>\frac12$.
\end{proposition}
\begin{proof}
Let the objects be $a,b,c,d$ and the preferences
$b\succ_1a\succ_1d\succ_1c$, $b\succ_2a\succ_2c\succ_2d$, $c\succ_3a\succ_3b\succ_3d$ and
$d\succ_4c\succ_4b\succ_4a$. Let
\[
P=\tfrac1{16}\begin{pmatrix}6&7&0&3\\6&7&3&0\\3&0&13&0\\1&2&0&13\end{pmatrix}
\quad(\text{columns } a,b,c,d),
\]
which is SD-EF. Exactly six assignments are compatible with the support of $P$:
$\sigma_1=(a,b,c,d)$, $\sigma_2=(b,a,c,d)$, $\sigma_3=(b,c,a,d)$, $\sigma_4=(d,a,c,b)$,
$\sigma_5=(d,b,c,a)$ and $\sigma_6=(d,c,a,b)$, where $\sigma=(x_1,\dots,x_4)$ means that agent $k$ receives $x_k$. The
decompositions of $P$ form the segment $\pi_s=\pi_0+s(0,1,-1,-1,0,1)$ with
$\pi_0=(\frac38,\frac7{24},\frac7{48},\frac1{12},\frac1{16},\frac1{24})$. The maximum-entropy
decomposition is the unique member of product form, characterized by
$\pi(\sigma_2)\pi(\sigma_6)=\pi(\sigma_3)\pi(\sigma_4)$, which is $\pi_0$. Agent~1
envies agent~2 exactly in $\sigma_1,\sigma_4,\sigma_5$, so $e_{12}(\pi_s)=\frac{25}{48}-s$ for
$s\in[-\frac1{24},\frac1{12}]$. In particular, $e_{12}(\pi_0)=\frac{25}{48}>\frac12$, while
$s=\frac1{12}$ gives $\tstar(P)=\frac7{16}$ (confirmed by an exact LP).
\end{proof}
SD-EF constrains $P$, not the likelihood ratios $W_{ia}W_{jb}/(W_{ib}W_{ja})$ that govern the
maximum-entropy decomposition. In this example the two agents outside the pair $\{1,2\}$ have
different rows, which is exactly the situation that Corollary~\ref{cor:maxent3} excludes. By
Theorems~\ref{thm:tworows} and~\ref{thm:twoplus}, maximum entropy is \decef for every four-agent
profile with at most three types, so a four-agent counterexample needs four distinct preferences,
as here.

\paragraph{Dependent rounding} A natural route to a uniform bound would compare the joint
probabilities $\Pr[\sigma(i)=a,\sigma(j)=b]$ with $P_{ia}P_{jb}$, as for independent draws. No such
comparison holds within a constant factor. Let $C$ be the permutation matrix of the cyclic shift
$i\mapsto i+1$. For $\varepsilon\le\frac12$, the matrix $(1-\varepsilon)I+\varepsilon C$ is SD-EF
whenever each agent $i$ ranks object $i$ first and object $i+1$ second. It has a unique
decomposition, and in it these ratios equal $1/\varepsilon$.
Dependent rounding~\cite{GKPS06}, whose negative-correlation guarantees concern edges at a common
vertex rather than pairs of disjoint edges, also violates \decef in our experiments, reaching $0.52$
at $n=5$ (Appendix~\ref{app:examples}).

\section{Computational Complexity}\label{sec:complexity}

A designer may want to know whether an existing lottery, for instance one produced by a
constrained mechanism~\cite{BCKM13}, can be implemented without concentrating ex-post envy. Such a
lottery need not be SD-EF, and on SD-EF inputs the question is trivial if Conjecture~\ref{conj:main}
holds. For arbitrary random assignments it is hard, in parallel with the NP-completeness of deciding
whether a random assignment has an ex-post efficient decomposition~\cite{AMXY15}.

\begin{theorem}\label{thm:np}
Given a profile and a rational bistochastic matrix $P$, deciding whether $P$ is EF-decomposable is
strongly NP-complete. Hardness holds even if every entry of $P$ lies in $\{0,\frac16,\frac14,\frac13\}$ and $P$ is
block diagonal with $5\times5$ blocks. Moreover, computing $\tstar(P)$ within an
additive error $1/\mathrm{poly}(n)$ is NP-hard.
\end{theorem}
\begin{proof}
\emph{Membership.} EF-decomposability is feasibility of an LP with $2n^2-n$ constraints and $0/1$
coefficients, so a feasible basic solution has at most $2n^2-n$ permutations in its support, each
with a weight of polynomial encoding length. The support and the weights form a certificate.

\emph{Construction.} Let $G=(V,E)$ be an instance of 3-\textsc{Colouring}, which is
NP-complete~\cite{GJS76}. Every vertex $v$ gets a \emph{colour agent} $a_v$, four \emph{filler} agents $f_v^1,\dots,f_v^4$,
\emph{colour objects} $x_v^1,x_v^2,x_v^3$ and two objects $h_v^1,h_v^2$. We set $P(a_v,x_v^t)=\frac13$,
$P(f_v^q,x_v^t)=\frac16$, $P(f_v^q,h_v^s)=\frac14$, and all other entries to $0$. Every row and
column sums to one. Each filler ranks $h_v^1,h_v^2,x_v^1,x_v^2,x_v^3$ first, in this order. Writing
$X^t_{N(v)}$
for the colour-$t$ objects of $v$'s neighbours, $a_v$ ranks
\[
X^3_{N(v)}\succ x_v^3\succ X^2_{N(v)}\succ x_v^2\succ X^1_{N(v)}\succ x_v^1
\]
first. All remaining objects follow in a fixed order, below those listed. Every permutation in the
support of a decomposition gives each agent an object of its own block, and it gives $a_v$ some
$x_v^{c(v)}$, which defines a colouring $c$. For an edge $uv$, the agent $a_u$ envies $a_v$ iff
$c(v)\ge c(u)$, and vice versa. Hence, for every decomposition $\pi$,
\begin{equation}\label{eq:star}
e_{a_ua_v}(\pi)+e_{a_va_u}(\pi)=1+\Pr_\pi[c(u)=c(v)].
\end{equation}

\emph{Decomposable implies 3-colourable.} If $P$ is EF-decomposable, the left-hand side of
\eqref{eq:star} is at most $1$ for every edge. So no edge is monochromatic in any assignment in the
support, and every such assignment yields a proper 3-colouring.

\emph{3-colourable implies decomposable.} Let $\chi$ be a proper colouring. Draw $\tau\in S_3$
uniformly, give $a_v$ the object $x_v^{c(v)}$ with $c(v)=\tau(\chi(v))$, and, independently for each
$v$, give the four leftover objects of block $v$ to its fillers by a uniformly random bijection.
Since $c(v)$ is uniform, $a_v$ receives each $x_v^t$ with probability $\frac13$. Each $h_v^s$ is
always left over, so a filler receives it with probability $\frac14$, and $x_v^t$ is left over with
probability $\frac23$, so a filler receives it with probability $\frac16$. Thus this decomposes $P$.
Now fix an ordered pair $(i,j)$ with $i$ in block $v$. If $i$ is a filler and $j$ lies outside block
$v$, or $i=a_v$ and $j$ lies in a block $w\notin N(v)\cup\{v\}$, then $i$ ranks $j$'s object below
every object that $i$ can receive, so $i$ never envies $j$. If $i=a_v$ and $j=a_w$ with $vw\in E$, then $i$ envies
$j$ iff $c(w)>c(v)$, and $(c(v),c(w))$ is uniform over the six ordered pairs of distinct colours, so
the probability is $\frac12$. If $i=a_v$ and $j$ is a filler of $v$ or of a neighbour, envy requires
$j$ to hold a colour object, which has probability $3\cdot\frac16=\frac12$. If $i$ is a filler of $v$
and $j=a_v$, envy requires $i$ to hold a colour object, again with probability $\frac12$. Finally, two
fillers of $v$ have identical preferences and the law is invariant under swapping them, so each
envies the other with probability $\frac12$. All numbers in the construction are constants, which
gives strong NP-hardness.

\emph{Inapproximability.} In a no-instance, every assignment in the support has a monochromatic
edge, so some edge is monochromatic with probability at least $1/|E|$, and \eqref{eq:star} gives
$\tstar(P)\ge\frac12+\frac1{2|E|}$. In a yes-instance, $\tstar(P)=\frac12$.
\end{proof}
The gap in Theorem~\ref{thm:np} becomes a constant if one starts from the hardness of
3-colouring with perfect completeness and a constant fraction of violated edges~\cite{Pet94}.
Hardness persists even on SD-EF inputs once the max-envy objective is replaced by a weighted sum:

\begin{proposition}\label{prop:fas}
Given a profile, a bistochastic $P$, weights $\mu_{ij}\in\{0,1\}$ and a rational $K$, deciding whether
some decomposition of $P$
has $\sum_{i\ne j}\mu_{ij}e_{ij}(\pi)\le K$ is NP-complete. This holds even if all agents have the same
preference and $P=J/n$, which is SD-EF.
\end{proposition}
\begin{proof}
Membership follows as in Theorem~\ref{thm:np}. With identical preferences,
$\mathrm{cost}(\sigma)=\sum_{(i,j)}\mu_{ij}\env_{ij}(\sigma)$ counts the arcs of the digraph
$A=\{(i,j):\mu_{ij}=1\}$ that point backward in the order that $\sigma$ induces on the agents. We
reduce from \textsc{Feedback Arc Set}~\cite{Kar72}: given a digraph on $r$ vertices with a
non-empty arc set $A$ and an integer $k$, is there an order with at most $k$ backward arcs? Add
dummy agents up to $n=(r-1)|A|+1\ge r$, let all agents rank $o_1\succ\dots\succ o_n$, and set
$K=k+(r-1)|A|/n<k+1$. If some decomposition meets $K$, then some permutation
in its support has integer cost at most $K$, hence at most $k$. Conversely, take an order with at
most $k$ backward arcs, followed by the dummies, and give the element in position $\ell$ the object
$o_{1+((\ell+s)\bmod n)}$ for a uniformly random shift $s$. This decomposes $J/n$. At most $r-1$ of
the $n$ shifts split the real agents, each with cost at most $|A|$, and the others cost at most $k$,
so the expected cost is at most $K$.
\end{proof}
With identical preferences, the weighted envy of a single assignment is the total weight of its
backward arcs, so the pricing problem of column generation for $\tstar$ (Lemma~\ref{lem:dual})
contains the weighted linear ordering problem.

\section{Conclusion}
We have settled the four-agent case of the open problem of \citet{KSSY26}. Beyond four agents,
all SD-EF matrices for five agents with at most four preference types are EF-decomposable, as are
all PS matrices for five agents, and adversarial search found no counterexample for up to seven
agents. The structural lemmas explain why a computer search is
conclusive here. $\tstar$ is convex, so vertices suffice. Every certificate is a small rational
decomposition, so the proof is checkable. The envy budget shows that \decef can only fail through
\emph{concentration} of envy, never through its total amount.

Three directions stand out. First, a proof of Conjecture~\ref{conj:main} needs a decomposition
beyond greedy Birkhoff--von~Neumann and maximum entropy, which fail in general
(Section~\ref{sec:canonical}), while symmetrization works when the row structure is coarse
(Theorem~\ref{thm:tworows}). The least-squares technique behind Theorem~\ref{thm:twoplus} is a
candidate route to profiles with more types. It would have to handle pair laws that are not of
product form, which arise once the agents outside a pair have two or more distinct rows. Second,
the bound $\tstar(P)\le\frac{n-1}n$ of \citet{KSSY26} tends to $1$, and no constant $\alpha<1$
independent of $n$ is known. A polynomial-time rounding that achieves such a constant would give a
usable guarantee without the full conjecture. Third, on SD-EF inputs the complexity of computing
$\tstar$ and of \emph{finding} a \decef decomposition is open. Our hardness results for general
inputs and for weighted objectives settle neither question.

\bibliographystyle{ACM-Reference-Format}
\bibliography{references}

\appendix
\section{Supplementary Proofs}\label{app:proofs}

\subsection{Two topological facts}\label{app:topo}

\begin{lemma}\label{lem:closed}
For every profile, the set $\{P:\tstar(P)\le\frac12\}$ of EF-decomposable matrices is a polytope, and in
particular it is closed.
\end{lemma}
\begin{proof}
It is the image of the polytope $\{\pi\in\mathbb{R}^{\Perm}_{\ge0}:\sum_\sigma\pi(\sigma)=1,\
\sum_\sigma\pi(\sigma)\env_{ij}(\sigma)\le\frac12\ \forall i\ne j\}$ under the linear map
$\pi\mapsto\sum_\sigma\pi(\sigma)Q^\sigma$. A linear image of a polytope is a polytope.
\end{proof}

\begin{lemma}\label{lem:maxentcont}
For bistochastic $P$, let $\pi^*(P)$ be the unique maximizer of the entropy
$H(\pi)=-\sum_\sigma\pi(\sigma)\log\pi(\sigma)$ over the decompositions of $P$. The map
$P\mapsto\pi^*(P)$ is continuous on the Birkhoff polytope.
\end{lemma}
\begin{proof}
Consider the feasible-set correspondence
\[F(P)=\{\pi\ge0:\textstyle\sum_\sigma\pi(\sigma)Q^\sigma=P\}.\]
It has non-empty compact values and a closed graph, so it is upper hemicontinuous. It is also lower
hemicontinuous. By Hoffman's error bound~\cite{Hof52}, there is a constant $\kappa$ with the following property.
For every $\pi\in F(P)$ and every bistochastic $P'$, some $\pi'\in F(P')$ satisfies
$\|\pi-\pi'\|\le\kappa\|P-P'\|$. The entropy is
continuous and strictly concave, so by Berge's maximum theorem~\cite{Ber63} the argmax correspondence
is upper hemicontinuous and single-valued, hence continuous.
\end{proof}

Since every $e_{ij}$ is linear in $\pi$, Lemma~\ref{lem:maxentcont} implies the following. If
$P_\varepsilon\to P$ and the maximum-entropy decompositions of all $P_\varepsilon$ are $\alpha$-\decef,
then so is that of $P$. We use this in two places:
\begin{itemize}
\item with $P_\varepsilon=(1-\varepsilon)P+\varepsilon J/n$ in Corollary~\ref{cor:maxent3};
\item with the same $P_\varepsilon$ in Theorem~\ref{thm:twoplus}.
\end{itemize}
Both properties that matter are preserved: $P_\varepsilon$ is SD-EF for the same profile, since the SD-EF
matrices form a convex set containing $J/n$, and agents with identical rows in $P$ keep identical rows.

\subsection{The maximum-entropy form (Corollary~\ref{cor:maxent3})}

Let $P>0$. The dual of entropy maximization shows $\pi^*(\sigma)\propto\prod_kW_{k\sigma(k)}$ with $W>0$,
where $\log W$ is any optimal solution of the convex dual
$\min_Y\log\sum_\sigma\exp(\sum_kY_{k\sigma(k)})-\langle Y,P\rangle$. If all agents in $S=\N\setminus\{i,j\}$ have
identical rows in $P$, the dual objective is invariant under permuting the rows of $Y$ indexed by $S$.
Averaging an optimal $Y$ over these permutations gives an optimal $Y$ whose $S$-rows coincide, by
convexity. The primal optimum $\pi^*$ is unique, so we may assume $W_k=w$ for all $k\in S$. Summing over
the assignments of $S$ to $\Ob\setminus\{a,b\}$ gives
\[
\begin{aligned}
\Pr_{\pi^*}[\sigma(i)=a,\sigma(j)=b]&\ \propto\ W_{ia}W_{jb}\,|S|!\prod_{c\ne a,b}w_c\\
&\ \propto\ \frac{W_{ia}}{w_a}\cdot\frac{W_{jb}}{w_b}\qquad(a\neq b),
\end{aligned}
\]
which is a zero-diagonal product-form coupling of $(P_i,P_j)$. SD-EF gives $P_i\succeq_{SD}P_j$ with respect
to $\succ_i$, so Lemma~\ref{lem:product} yields $e_{ij}\le\frac12$; symmetrically $e_{ji}\le\frac12$. For
general $P$, apply Lemma~\ref{lem:maxentcont}. For $n=3$ the complement of any pair is a single agent.

\subsection{The product-form coupling (Theorem~\ref{thm:twoplus})}

Let $p=P_u$, $q=P_v$, and let $r$ be the common row of the $m=n-2$ agents in $A=\N\setminus\{u,v\}$, and
assume $P>0$.

\emph{Existence.} Write $P=(1-\eta)P'+\eta J/n$ with $P'$ bistochastic and $\eta>0$ small. Mixing any
decomposition of $P'$ with the uniform distribution over $\Perm$ gives a decomposition of $P$ with full
support. The induced law of $(\sigma(u),\sigma(v))$ is a coupling of $(p,q)$ that is positive on every
off-diagonal cell. By the matrix-scaling theorem for prescribed zero patterns~\cite{SK67,Men68,Bru68},
there are $U,V>0$ such that $Q_{ab}=U_aV_b$ ($a\ne b$), $Q_{aa}=0$, has row sums $p$ and column sums
$q$. It is computed by alternating row and column scaling.

\emph{Uniqueness and maximum entropy.} By the previous subsection (with $S=A$), the maximum-entropy
decomposition has a product-form $(u,v)$-marginal $Q'$ with marginals $(p,q)$. A zero-diagonal
product-form coupling with given positive marginals is unique, because it is the entropy maximizer
among couplings with that zero pattern. So $Q'=Q$. Given $(\sigma(u),\sigma(v))=(a,b)$, the conditional law
of the assignment of $A$ is proportional to $\prod_{x\in A}w_{\sigma(x)}=\prod_{o\ne a,b}w_o$, which is uniform.
The construction was also checked end to end by explicit enumeration (script
\texttt{construction\_\allowbreak check.py}): on $450$ random SD-EF instances with $n=4,5,6$, its marginals equal
$P$ up to $10^{-13}$ and every envy probability is at most $\frac12$.

\subsection{Remarks on Lemma~\ref{lem:ls}}\label{app:lsremarks}

\emph{The unit-vector case.} Lemma~\ref{lem:ls} excludes $V=e_m$, which cannot occur in Theorem~\ref{thm:twoplus} because there $V>0$. If $V=e_m$, then $k_{mb}=0$ for all $b$
and $s_m=0$, so $\psi_m$ is free. On $\Ob'=\Ob\setminus\{m\}$ all weights equal $1$ and
$s_b=n-2r_b$, where $r_b$ is the rank of $b$ in $\Ob'$. The solution with mean zero on $\Ob'$ is
$\psi_b=s_b/(n-1)$, which is strictly decreasing, and choosing $\psi_{m+1}\le\psi_m\le\psi_{m-1}$ gives a
monotone solution. The minimum-norm solution, which sets $\psi_m=0$, is in general not monotone.

\emph{A quantitative by-product.} The proof of Lemma~\ref{lem:ls} shows $y\in[0,1]^n$, hence
$0\le\psi_a-\psi_{a+1}\le\frac4n$ for the solution with $\sum_a\psi_a=0$.

\emph{The sign condition is essential.} Both uses of $\nu\ge0$ in inequality~\eqref{eq:key} are
needed. For general additive weights $g_a+g_c$ with $g\ge0$, which corresponds to allowing $\nu_a<0$, we
found explicit instances in which the least-squares ranking is not monotone. For example, $n=6$ and
$g=(0,0,0,\frac1{10},0,2)$ give a connected weight graph and, in exact arithmetic,
$\psi_5-\psi_4=\frac{93}{3647}>0$.

\emph{Verification.} All identities in the proof were also verified symbolically (sympy,
$n=3,\dots,7$), and the conclusion was checked in exact rational arithmetic on $4{,}836$ instances covering
both cases of the proof and the unit-vector case (script \texttt{lemma\_\allowbreak proof\_\allowbreak check.py}).

\subsection{The two-row LP (Section~\ref{sec:canonical})}\label{app:tworowlp}

For a fixed bipartition $(A,B)$ of the agents, the maximal weight $\mu$ with which $P$ contains a
component that has rows $r_A$ on $A$ and $r_B$ on $B$ and satisfies the hypothesis of Theorem~\ref{thm:tworows} is an LP. Its variables are the scaled rows $\tilde r_A=\mu r_A$,
$\tilde r_B=\mu r_B$ and the weight $\mu$. Its constraints are:
\begin{itemize}
\item $\tilde r_A,\tilde r_B\ge0$, with row sums $\mu$ and $|A|\tilde r_A+|B|\tilde r_B=\mu\mathbf 1$;
\item the expected-rank constraints $\sum_o\tilde r_{A,o}\rk_i(o)\le\mu\frac{n-1}2$ for $i\in A$, and likewise
for $B$;
\item $\tilde r_A\le P_i$ entrywise for $i\in A$, and likewise for $B$.
\end{itemize}

\section{Computational Details}\label{app:comp}

\subsection{Pipeline and what is certified}

For a profile $R$, the SD-EF polytope $\sdef(R)$ has the following H-representation:
\begin{itemize}
\item $n^2$ non-negativity constraints;
\item $2n$ equalities (row and column sums);
\item $n(n-1)^2$ prefix inequalities.
\end{itemize}
Its vertices are enumerated with the double-description method of cddlib~\cite{cddlib} in exact GMP
rational arithmetic. For each vertex $P$ we solve
\[
\min\ t\quad\text{s.t.}\quad \textstyle\sum_\sigma\pi_\sigma Q^\sigma=P,\ \ \sum_\sigma\pi_\sigma\env_{ij}(\sigma)\le t\ (i\ne j),\ \ \pi\ge0
\]
in floating point (HiGHS). If the value exceeds $\frac12-10^{-7}$, we re-solve the LP \emph{exactly},
restricted to the support of the floating-point solution. If that fails, we solve the full LP exactly. The
output is a rational $\pi$ satisfying all constraints with $t\le\frac12$. In the type-reduced runs every
vertex exceeds the threshold, because agents with identical preferences force $\tstar=\frac12$ (Corollary~\ref{cor:half}), and for $n=4$ a separate pass certifies every vertex. For $n=4$ we store all $26{,}927$ certificates, and an
independent $40$-line checker verifies each one using Python's \texttt{Fraction} only:
\begin{itemize}
\item marginals equal $P$;
\item $P$ is SD-EF;
\item every envy probability is at most $\frac12$;
\item the stored orbits cover all $24^4$ profiles.
\end{itemize}
For $n\ge5$ the certificates are regenerated on demand rather than stored.

\emph{Type reduction.} If $\succ_i=\succ_j$, SD-EF in both directions gives $P_i(U^k)=P_j(U^k)$ for every prefix
$U^k$, hence $P_i=P_j$. So $\sdef(R)$ is affinely isomorphic to a polytope with one row per distinct
preference and column constraints $\sum_t m_tP_t=\mathbf 1$, where $m_t$ are the multiplicities. This is what
makes $n=5$ with four types and $n=6$ with pattern $(2,2,2)$ tractable.

\emph{Orbits.} Profiles are enumerated up to renaming agents and objects. A profile is encoded by its
multiset of preference indices after relabeling objects so that one of its preferences becomes the identity.
The canonical form is the lexicographically smallest such encoding. The counts are
$10$ ($n=3$), $762$ ($n=4$) and $1{,}876{,}255$ ($n=5$). The orbit counts per multiplicity pattern are
$7{,}021$ for $(3,1,1)$ and for $(2,2,1)$, $273{,}819$ for $(2,1,1,1)$ ($n=5$), and $86{,}067$ for $(2,2,2)$ and $258{,}121$ for $(4,1,1)$ ($n=6$).

\subsection{Runs}

\begin{table*}[t]
\caption{Exact runs. Every vertex (for PS, every PS matrix) received a rational certificate of
$\tstar\le\frac12$. The last row is a random sample. The pattern $(4,1,1)$ is covered for every $n$ by
Theorem~\ref{thm:twoplus}.}
\small
\begin{tabular}{@{}lrrrl@{}}
\toprule
class & orbits & vertices & max/orbit & hardware, time\\
\midrule
$n=4$, all & 762 & 26{,}927 & 375 & 24 cores, $<1$ min\\
$n=5$, $(3,1,1)$ & 7{,}021 & 661{,}039 & 403 & 6 procs, 16 min\\
$n=5$, $(2,2,1)$ & 7{,}021 & 937{,}832 & 529 & 6 procs, 19 min\\
$n=5$, $(2,1,1,1)$ & 273{,}819 & 213{,}980{,}202 & 11{,}910 & 96+8 procs, $\approx$1 day\\
$n=6$, $(2,2,2)$ & 86{,}067 & 50{,}833{,}539 & 3{,}140 & 90 procs, 7 h\\
$n=5$, PS & 1{,}876{,}016 & -- & -- & 10 procs, 8 min\\
$n=6$, $(4,1,1)$, sample & 300 of 258{,}121 & 148{,}771 & 1{,}833 & 6 min\\
\bottomrule
\end{tabular}
\end{table*}

\emph{Distributed run.} For pattern $(2,1,1,1)$, the orbit list was generated once and split into
$2{,}709$ chunks of $100$ orbits. The $3{,}000$ orbits already checked in a preliminary random sample were
excluded and counted separately. Workers on a $96$-core server and on a workstation claimed chunks through
an atomic \texttt{mkdir} in a shared claim directory: the server claimed in ascending order and the
workstation in descending order. A final merge confirmed that every chunk was completed exactly once.
One chunk, whose worker was interrupted, was recomputed.

\emph{Exact PS run.} The run behind Theorem~\ref{thm:ps5} (script
\texttt{certify\_\allowbreak rule\_\allowbreak exact.py}) certifies all $1{,}876{,}016$ orbits exactly. It
re-solves the LP in exact arithmetic on the support of the floating-point solution (the fallback to the
full exact LP was never needed) and re-checks each rational decomposition with Python's
\texttt{Fraction}: nonnegative weights, marginals equal to the exact PS matrix, and every envy
probability at most $\frac12$. All certificates pass. An earlier floating-point screen (script
\texttt{run\_\allowbreak rule\_\allowbreak exhaustive.py}) had already certified exactly the
$1{,}561{,}882$ orbits whose floating-point value exceeds $\frac12-10^{-7}$. Among the other
$314{,}134$ orbits (floating-point value at most $\frac12-10^{-7}$), the largest exact $\tstar$ is
$\frac{40}{81}$.

\emph{Adversarial ascent.} For random profiles we start from a random vertex of $\sdef(R)$, maximizing a
random linear objective. We then alternate between (i) solving the min-max LP and reading off the dual
multipliers $Y$ of the marginal constraints, and (ii) moving to the vertex of $\sdef(R)$ maximizing $\langle Y,\cdot\rangle$.
By Lemma~\ref{lem:dual}, the value never decreases. Runs:
\begin{itemize}
\item $n=5$: $20{,}000$ random profiles and $5{,}000$ three-type profiles, $30$ restarts each;
\item $n=6$: $3{,}000$ random profiles and $2{,}000$ three-type profiles, $20$ restarts each;
\item $n=7$: $1{,}300$ random profiles, $10$ restarts each.
\end{itemize}
The maximum was $\frac12$ up to floating-point error in all of them.

\subsection{Examples}\label{app:examples}

\emph{Greedy Birkhoff--von~Neumann (Section~\ref{sec:canonical}).} Objects $a,b,c,d$; preferences
$a\succ_1b\succ_1c\succ_1d$, $a\succ_2c\succ_2d\succ_2b$, $c\succ_3b\succ_3d\succ_3a$, $d\succ_4a\succ_4b\succ_4c$, and
\[
P=\tfrac1{28}\begin{pmatrix}9&8&8&3\\9&4&8&7\\5&4&12&7\\5&12&0&11\end{pmatrix}.
\]
$P$ is SD-EF. Repeatedly removing a maximum-bottleneck permutation (our implementation breaks ties by the
assignment solver) yields the decomposition
$bacd\,(\frac8{28})$, $acdb\,(\frac7{28})$, $cdab\,(\frac5{28})$, $dbca\,(\frac3{28})$, $cdba\,(\frac2{28})$,
$abcd$, $acbd$, $cabd$ ($\frac1{28}$ each), where $xyzw$ means that agents $1,\dots,4$ receive $x,y,z,w$. In it, agent~2 envies agent~1 with probability
$\frac{19}{28}$, whereas $\tstar(P)=\frac{25}{56}$ (exact LP).

\emph{Maximum entropy (Proposition~\ref{obs:natural}).} The example is reproduced by
\texttt{examples.py}, which checks exactly that $P$ is SD-EF, that $\pi_0$ has marginals $P$, that
$\pi_0(\sigma_2)\pi_0(\sigma_6)=\pi_0(\sigma_3)\pi_0(\sigma_4)$, that $e_{12}(\pi_0)=\frac{25}{48}$, and that $\tstar(P)=\frac7{16}$.

\emph{Other natural decompositions.} Bipartite dependent rounding~\cite{GKPS06} with a deterministic
(lexicographic) choice of cycles reaches maximum envy $0.522\pm0.004$ on an $n=5$ SD-EF instance ($60{,}000$
samples, 95\% interval). The version with random cycle choice reaches about $0.51$. A random search found maximum-entropy
envy up to $0.502$ at $n=5$.

\section{Reproducibility}\label{app:repro}
\begingroup\raggedright\sloppy

All code and result summaries are provided as ancillary files of this arXiv submission (directory \texttt{anc/code/}). It requires Python~3.10+,
\texttt{numpy}, \texttt{scipy} (HiGHS) and \texttt{pycddlib}~3 (cddlib with GMP). Main entry points:
\begin{itemize}
\item \texttt{run\_\allowbreak exhaustive.py 4}, \texttt{certify\_\allowbreak n4.py} and \texttt{verify\_\allowbreak certificates.py}: Theorem~\ref{thm:n4};
\item \texttt{run\_\allowbreak typed\_\allowbreak exhaustive.py 5 3,1,1} (and \texttt{2,2,1}), and \texttt{chunk\_\allowbreak worker.py} for
$(2,1,1,1)$: Theorem~\ref{thm:n5types};
\item \texttt{run\_\allowbreak typed\_\allowbreak exhaustive.py 6 2,2,2 -{}-part p -{}-parts 12}: Theorem~\ref{thm:n6types};
\item \texttt{run\_\allowbreak rule\_\allowbreak exhaustive.py 5} (floating-point screen) and \texttt{certify\_\allowbreak rule\_\allowbreak exact.py 5}
(exact certificate for every orbit): Theorem~\ref{thm:ps5};
\item \texttt{ascent.py}: adversarial search;
\item \texttt{examples.py}: examples;
\item \texttt{theory\_\allowbreak approx/gkps.py} and \texttt{theory\_\allowbreak approx/run\_\allowbreak gkps2.py} (instance in
\texttt{gkps\_\allowbreak lex\_\allowbreak example.json}): dependent rounding;
\item \texttt{theory\_\allowbreak twoplus/lemma\_\allowbreak proof\_\allowbreak check.py}: symbolic and exact checks for Lemma~\ref{lem:ls};
\item \texttt{theory\_\allowbreak twoplus/construction\_\allowbreak check.py}: end-to-end check of Theorem~\ref{thm:twoplus};
\item \texttt{theory\_\allowbreak complexity/}: numerical sanity checks of the reductions of Section~\ref{sec:complexity}.
\end{itemize}
\endgroup

\end{document}